\documentclass[12pt]{article}

\usepackage{amsmath,amsfonts,amssymb,amsthm,textcomp,mathrsfs,mathrsfs,bbm}
\usepackage{xfrac,braket}
\usepackage{mathtools}

\newtheorem{theorem}{Theorem}
\newtheorem{lemma}[theorem]{Lemma}
\newtheorem{definition}{Definition}

\newcommand*{\cC}{\mathcal{C}}
\newcommand*{\cM}{\mathcal{M}}
\newcommand*{\brackett}[3]{\left\langle #1 \right| \left. #2 \right. \left| #3 \right\rangle}
\newcommand*{\cD}{\mathcal{D}}
\newcommand*{\eye}{{\mathbbm{1}}}
\newcommand*{\bbR}{{\mathbb{R}}}
\newcommand*{\cE}{\mathcal{E}}
\newcommand{\beq}{\begin{equation}}
\newcommand{\enq}{\end{equation}}
\newcommand{\ketbra}[1]{\ket{#1} \bra{#1}}
\newcommand*{\sep}{{\mathscr{S}}}
\newcommand{\tr}{\mathrm{Tr}}
\newcommand*{\cS}{\mathcal{S}}
\newcommand*{\cH}{\mathcal{H}}
\newcommand*{\cI}{\mathcal{I}}
\newcommand*{\renyi}{R\'{e}nyi }

\newcommand{\ere}[1]{E_{RE}^{(#1)}}

\usepackage[colorlinks=true,linkcolor=black,citecolor=black,plainpages=false,pdfpagelabels]{hyperref}
\hypersetup{pdftitle={Template}}
\usepackage{color}

\usepackage{fullpage}

\begin{document}

\newcommand{\itwomax}{{^2I}_{\max}}

\title{A strong converse for the quantum state merging protocol}
\author{Naresh Sharma \\
Tata Institute of Fundamental Research \\
Mumbai, India \\
Email: {\tt nsharma@tifr.res.in}}
\date{\today}
\maketitle

\begin{abstract}
The Polyanskiy-Verd\'{u} paradigm provides an elegant way of using
generalized-divergences to obtain strong converses and thus far has remained
confined to protocols involving channels (classical or quantum). In this paper,
drawing inspirations from it, we provide strong converses for protocols
involving LOCC (local
operations and classical communication). The key quantity
that we work with is the \renyi relative entropy of entanglement.
We provide a strong converse for the quantum state merging
protocol that
gives an exponential decay of the fidelity of the protocol for rates below the
optimum with the number of copies of the state and are provided both for
entanglement rate with LOCC as well as for
classical communication with one-way LOCC. As an aside, the developments
also yield short strong converses for the entanglement-concentration
of  pure states and the Schumacher compression.
\end{abstract}

\section{Introduction}

Many information-theoretic problems deal with finding out the
minimal resources needed to accomplish a task or
the maximal yield obtained after a task given the
constraints and many times a definite optimal
answer is provided when the number of copies of the input(s),
resources, or output(s)
is large and examples include Schumacher compression, quantum
state merging, entanglement concentration etc.
\cite{covertom, wilde-book, hayashi, petz-book}.

These answers are typically given in two parts: for the number of
resources larger (respectively yield smaller) than the optimal, there exists
a protocol to accomplish
the task (termed as achievability) and for the number of resources
smaller (respectively yield larger) than the optimal, any protocol will
perform badly (termed as converse). In the latter case, a strong
converse, if it exists, additionally says that
any protocol will perform very badly (as bad as it can be)
and this is quantified in terms
of a performance measure. Strong converses provide a refined
view to the optimal quantities since now they can be
seen as a sharp dividing line between what can and
cannot be achieved.

These strong converses have a long history starting from the
works in the classical case of Wolfowitz \cite{wolfowitz-book},
Arimoto \cite{arimoto-1973-converse}
to more recent works in the quantum case of
Winter \cite{winter-99-converse, andreas-thesis}, Ogawa \&
Nagaoka \cite{ogawa-1999-converse} and K\"{o}nig and Wehner
\cite{konig-2009-converse}. This is hardly a list of exhaustive
references on the topic since there have been more recent works (some
of which are referenced later) and the literature on smooth
entropies starting from the work of Renner and Wolf \cite{renner-wolf-2004}
provides bounds that when coupled with asymptotic equipartition
property for the independent and identically distributed (i.i.d.) copies
yield strong converses as well.

Recently, a fresh and enchanting  take on strong converses has been
provided by Polyanskiy and Verd\'{u}
using the generalized relative entropies that satisfy certain well
expected properties such as monotonicity under the application of
the classical channels \cite{polyanskiy-2010-converse}.

The idea of using monotonicity to prove converses is due to Blahut who used it
to prove the Fano inequality that gives the weak converse \cite{blahut1976}.
Blahut employed the traditional relative entropy as opposed to the generalized
relative entropies employed by Polyanskiy and Verd\'{u} who show that,
in particular, the \renyi relative entropy would yield the Arimoto converse.

Arimoto's proof was extended by Ogawa \& Nagaoka for
sending classical information across quantum channels. But it has not
been possible to extend Arimoto's proof for other protocols such as
getting an exponential bound for the quantum information transfer
across quantum channels. 
Unlike Arimoto's proof, Polyanskiy and Verd\'{u}'s proof relied on certain
properties of generalized relative entropies that are also satisfied
by the quantum generalized relative entropies and their approach
has now been extended to the quantum
domain (see Ref. \cite{sharma-prl-2013}) and applied to various protocols
\cite{wilde-2013,gupta-wilde-2013} (see also Ref. \cite{wilde-2013} for a
discussion on this).

Not all protocols admit strong converses \cite{no-strong-converse-2011}
and in some cases, no definite answer is known whether a strong
converse would exist or not. For example, we don't have a strong
converse for the quantum capacity except in
some special cases mentioned in Ref. \cite{morgan-winter-2014} although
the same paper provides a ``pretty strong converse" for degradable channels.

The Polyanskiy-Verd\'{u} paradigm has thus far remained confined
to protocols involving channels. Quantum information theory is richly
endowed with another class of operations namely the LOCC.
Inspired by the paradigm, could one address the protocols involving LOCC?

Just as Polyanskiy and Verd\'{u} define \renyi mutual information that does not
increase under the application of the channel (see the quantum \renyi
information measures in Ref. \cite{sharma-prl-2013}), we seek
a \renyi entanglement measure that does not increase under LOCC.
The relative entropy of entanglement
is an entanglement monotone \cite{vidal-2000, entanglement-horo-2009}
and hence, does not increase under LOCC. A \renyi relative entropy of
entanglement is then not hard to define that, leveraging the result by
Vedral \emph{et al} \cite{vedral-1997}, does not increase under LOCC.

For the protocols involving channels,
typically, there is a reduction, using monotonicity, to \renyi relative entropy
involving binary distributions that are a functions of probability of
error or the fidelity that tell us how well the protocol performs.
Such reductions no longer seem to be applicable/useful for protocols
such as the quantum state merging. Instead an inequality by
van Dam and Hayden \cite{dam-hayden-2008}
involving the \renyi entropies of two states 
and their fidelity given in a completely different context
turns out to be just the thing one is looking for and it is also
useful in providing bounds for the \renyi relative entropy
of entanglement.

For a Hilbert space $\cH$, we define $\cS(\cH) = \{\rho \geq 0 : \tr \rho = 1\}$.
Fidelity $F$ between the states $\rho$, $\sigma$ $\in \cS(\cH)$ is
$F(\rho,\sigma) \equiv || \sqrt{\rho} \sqrt{\sigma} ||_1 = $ $\max| \bra{\phi}
\ket{\psi} |$,
where the maximization is over all purifications $\ket{\phi}$, $\ket{\psi}$
of $\rho$ and $\sigma$ respectively.

We denote a maximally entangled state in $\cH_A \otimes \cH_B$
with Schmidt rank $K$ by $\Phi_K^{AB}$.
The set of bipartite separable states of $A$ and $B$
is denoted by $\sep^{A:B}$. The terms quantum operation and completely 
positive trace preserving (cptp) map are used interchangeably.
For a pure state $\ket{\Psi}$, $\Psi = \ketbra{\Psi}$.

All the logarithms are to the base $2$.
The quantum relative entropy from $\rho$ to $\sigma$ is given by
$S(\rho || \sigma) \equiv \tr \rho (\log \rho - \log \sigma)$, the von Neumann
entropy is given by $S(A)_\rho \equiv - \tr \rho^A \log \rho^A$.
At times, we shall use $S_\alpha(\rho^A)$ instead of $S_\alpha(A)_\rho$.
For a bipartite state $\rho^{AB}$, the conditional entropy of $A$ given $B$
is given by $S(A|B)_\rho \equiv S(AB)_\rho - S(B)_\rho$ and the
quantum mutual information between $A$ and $B$ is given by
$I(A:B)_\rho \equiv S(A)_\rho + S(B)_\rho - S(AB)_\rho$.

For $\rho$, $\sigma$ $\geq 0$, $\alpha \in [0,2] \backslash \{1\}$,
the $\alpha$-quasi-relative entropy,
from $\rho$ to $\sigma$ ($\rho$, $\sigma$ $\geq 0$) is defined as
\beq
Q_\alpha(\rho || \sigma) \equiv
\text{sign}(\alpha-1) \tr \rho^\alpha \sigma^{1-\alpha},
\enq
and the \renyi $\alpha$-relative entropy from $\rho$ to $\sigma$ is defined as
\beq
S_\alpha(\rho || \sigma) \equiv
\frac{1}{\alpha-1} \log \tr \rho^\alpha \sigma^{1-\alpha}
\enq
%\beq
%S_\alpha(\rho || \sigma) \equiv \left\{
%\begin{array}{ll}
%\frac{1}{\alpha-1} \log \tr \rho^\alpha \sigma^{1-\alpha}, &
%\text{supp} ~ \rho \leq \text{supp} ~ \sigma \\
%& \text{or } \alpha \in [0,1) \\
%+ \infty, & \text{otherwise}
%\end{array}
%\right.
%\enq
where limits are taken for $\alpha = 1$ and we drop the subscript.
One can derive this from the generalized divergences defined by Petz
\cite{petz-quasi-entr-1986}.
We quickly recall from Refs. \cite{petz-quasi-entr-1986, petz-quasi-entr-2010}
that for $\alpha \in [0,2]$ and a cptp map $\cE$,
\beq
S_\alpha(\rho || \sigma) \geq S_\alpha\left[ \cE(\rho) || \cE(\sigma)\right].
\enq

Let $S_{\alpha}(A)_\rho$, $\alpha \geq 0$ be the $\alpha$-entropy of
$\rho^A$ given by
\beq
S_{\alpha}(A)_\rho = \frac{1}{1-\alpha} \log \tr (\rho^A)^{\alpha},
\enq
where, again, limits are taken for $\alpha = 1$ and we drop
the subscript.
The \renyi coherent information (see Ref. \cite{sharma-prl-2013})
is given by
\beq
I_\alpha(A \rangle B)_\rho \equiv
\frac{\alpha}{\alpha-1} \log \tr \left[ \tr_B (\rho^{AB})^\alpha \right]^{1/\alpha}.
\enq

\section{Nature of the bounds obtained}

The strong converse bounds we provide are
obtained using the \renyi relative entropy similar in spirit to
that of Arimoto and Polyanskiy \& Verd\'{u}
\cite{arimoto-1973-converse, polyanskiy-2010-converse}.
These bound are provided for the quantum state merging,
entanglement concentration and the Schumacher compression.
We note that strong converse is already known for these
protocols but, to the best of author's knowledge, there are no
strong converses known using the \renyi approach.

We now illustrate the nature of the bounds we provide.
Suppose for a bipartite state $\rho^{AB}$,
the optimal quantity is given in terms of a function $f(\rho^{AB})$.
Let $\cI \subseteq \bbR$ be an interval with $1$ as its boundary point
and not containing $\{0\}$.
Suppose a \renyi generalization of  $f(\rho^{AB})$ for
$\alpha \in \cI$ is given
by $f_\alpha(\rho^{AB})$ with  $f_\alpha(\rho^{AB})$ $\leq$ $f(\rho^{AB})$
for $\alpha \in \cI$,
where $\lim_{\alpha \to 1} f_\alpha(\rho^{AB}) = f(\rho^{AB})$ and
for $n$ copies,
$f_\alpha\left[ (\rho^{AB})^{\otimes n} \right]$
$= n f_\alpha(\rho^{AB})$. Suppose the resources consumed that we want
to lower bound are $g(n)$ for $n$ copies.
The bounds that we obtain are of the form
\beq
\label{gen-bound}
\log(\text{Fidelity of the protocol}) \leq n \zeta \left| \frac{\alpha-1}{\alpha}
\right| \left[ \frac{g(n)}{n} - f_\alpha(\rho^{AB}) \right],
\enq
where $\zeta > 0$ is a constant not dependent on any parameters.
These bounds clearly have the same flavor
as the Arimoto converse \cite{arimoto-1973-converse}. If for all $n$,
$g(n)/n$ is bounded from above by $R$ such that $R < f(\rho^{AB})$,
then it follows that we can choose a $\alpha$ close to $1$ such that
$g(n)/n - f_\alpha(\rho^{AB})$ $\leq R - f_\alpha(\rho^{AB})$ is
negative and the RHS is independent of $n$.

Note that in some cases,
instead of resources consumed we are interested in the yield
$h(n)$ (for $n$ copies). For example, in the case of entanglement
concentration, we are interested in the number
of EPR pairs generated. Then the bounds obtained are of the form
\beq
\label{gen-bound2}
\log(\text{Fidelity of the protocol}) \leq n \zeta \left| \frac{\alpha-1}{\alpha}
\right| \left[ f_\alpha(\rho^{AB}) - \frac{h(n)}{n} \right].
\enq
In this case, $f_\alpha(\rho^{AB}) \geq f(\rho^{AB})$ with
$\lim_{\alpha \to 1} f_\alpha(\rho^{AB}) = f(\rho^{AB})$. If $h(n)/n$
$ \geq R$ and $R > f(\rho^{AB})$, then $f_\alpha(\rho^{AB}) - h(n)/n$
$\leq f_\alpha(\rho^{AB}) - R$ and hence, we can choose an $\alpha$
such that $f_\alpha(\rho^{AB}) - R$ is negative and the RHS is
independent of $n$.

In either case, the fidelity decays exponentially with $n$. Our
purpose of stating this `last mile' common to several protocols
is to avoid repetition and we shall henceforth state a strong converse as
the bounds in \eqref{gen-bound} or \eqref{gen-bound2}.

\section{\renyi relative entropy of entanglement}

We define the \renyi $\alpha$-relative entropy of entanglement (RREE) for a 
bipartite state $\rho^{AB}$ as
\beq
\ere{\alpha}(A:B)_\rho \equiv \inf_{\sigma^{AB} \in \sep^{A:B}}
S_\alpha( \rho^{AB} || \sigma^{AB}).
\enq
We now prove some of its properties.

\begin{lemma}
\label{lemma1}
For any cptp map $\cE: B \to C$ and $\alpha \in [0,2]$,
\begin{align}
\ere{\alpha}(A:B)_\rho & \geq \ere{\alpha}(A:C)_{\cE(\rho)}.
\end{align}
(Note that the same applies to a local map over $A$ for $\ere{\alpha}$
as well due to symmetry.)
\end{lemma}
Proof follows from monotonicity using arguments similar to those in Refs.
\cite{polyanskiy-2010-converse, sharma-prl-2013} and we omit the details.

\begin{lemma}[Vedral \emph{et al} \cite{vedral-1997}]
RREE is LOCC-monotone, i.e., it does not increase under LOCC.
\end{lemma}

We now show that for $\alpha \in (1,2]$,
RREE satisfies a stronger condition than LOCC monotonicity
(see Sec XV.B.1 in Ref. \cite{entanglement-horo-2009} and also Refs.
\cite{vidal-2000, christandl-winter-2004})). More specifically, this
condition states that RREE does not increase on average under the
action of LOCC.

\begin{lemma}
For $\alpha \in (1,2]$, for any quantum state $\rho^{AB}$ and any unilocal
quantum instrument performed without loss of generality over subsystem $A$
($\cE_k : A \to A^\prime$) - the $\cE_k$ are completely
positive maps and their sum is trace preserving - and
orthogonal states $\{ \ketbra{k}^X \}$, and $\tau^{XA^\prime B} = \sum_x p_k \ketbra{k}^X
\otimes \theta_k^{A^\prime B}$, $p_k = \tr \cE_k(\rho^{AB})$, $\theta_k^{A^\prime B} = 
\cE_k(\rho^{AB})/p_k$, we have
\beq
\ere{\alpha}(A:B)_{\rho} \geq \ere{\alpha}(XA^\prime:B)_\tau \geq
\sum_k p_k \ere{\alpha}(A^\prime:B)_{\theta_k}.
\enq
\end{lemma}
\begin{proof}
The first inequality follows from Lemma
\ref{lemma1} where the local cptp map ($A \to X A^\prime$)
in question is $\sum_k \ket{k}^X \otimes
\cE_k$. To make the notation clear, let
$\cE_k(\rho^{AB}) = \sum_j E_{jk} \rho^{AB} E_{jk}^\dagger$, where
$\sum_{j,k} E_{jk}^\dagger E_{jk} = \eye$. Then the Kraus operators of the
above cptp map are $\{ \ket{k}^X \otimes E_{jk} \}$.

We now prove the second inequality.
Define a quantum operation $\cD: XA^\prime B \to XA^\prime B$ with the
Kraus operators $\{\ketbra{k}^X \otimes \eye^{A^\prime B}\}$. We note that the
output of $\cD$ for any input is a cq (classical quantum) state which is
classical on $X$ and $\cD$ does not alter the state $\tau^{XA^\prime B}$.
Furthermore, for any $\sigma^{XA^\prime B} \in \sep^{XA^\prime :B}$, if
$\cD(\sigma^{XA^\prime B})$
$=$ $\sum_k q_k \ketbra{k}^X \otimes \sigma_k^{A^\prime B}$, where $\{q_k\}$ is a probability
vector, then $\sigma_k^{A^\prime B} \in \sep^{A^\prime :B}$ for all $k$.
We now have
\begin{align}
\ere{\alpha} & (XA^\prime :B)_{\tau} \nonumber \\
& = \inf_{\sigma^{XA^\prime B} \in \sep^{XA^\prime :B}} S_\alpha(\tau^{XA^\prime B} ||
\sigma^{XA^\prime B}) \\
& \geq
\inf_{\sigma^{XA^\prime B} \in \sep^{XA^\prime :B}}
S_\alpha \left[ \cD(\tau^{XA^\prime B}) || \cD(\sigma^{XA^\prime B})
\right] \\
& \geq \inf_{ \{q_k\}, \{ \sigma_k^{A^\prime B} \in \sep^{A^\prime:B} \} }
S_\alpha \left[ \tau^{XA^\prime B} ||
\sum_k q_k \ketbra{k}^X \otimes \sigma_k^{A^\prime B} \right] \\
& = \inf_{ \{q_k\}, \{ \sigma_k^{A^\prime B} \in \sep^{A^\prime:B} \} } \frac{1}{\alpha-1} \log \left[
\sum_k p_k \left( \frac{p_k}{q_k} \right)^{\alpha-1}
Q_\alpha(\theta_k^{A^\prime B} || \sigma_k^{A^\prime B}) \right] \\
& \geq \inf_{ \{q_k\}, \{ \sigma_k^{A^\prime B} \in \sep^{A^\prime:B} \} }
\left[ S(\underline{p} || \underline{q})
+ \sum_k p_k S_\alpha (\theta_k^{A^\prime B} || \sigma_k^{A^\prime B}) \right] \\
& \geq \sum_k p_k \ere{\alpha}(A^\prime:B)_{\theta_k},
\end{align}
where the first inequality follows because of monotonicity under cptp maps,
the second inequality follows because we are minimizing over a bigger set,
the third inequality follows because of the concavity of the logarithm,
$S(\underline{p} || \underline{q}) = \sum_k p_k \log(p_k/q_k)$
is the classical relative entropy from probability vectors
$\underline{p}$ to $\underline{q}$, and is non-negative and zero if and only if
$\underline{p} = \underline{q}$ (see Ref. \cite{covertom}),
and the last inequality follows by choosing
the minimizing $\underline{q}$ and minimizing $\{ \sigma_k^{AB} \}$.
\end{proof}

We now generalize Theorem 4.7 in Ref. \cite{petz-statistics-book}.
\begin{lemma}
\label{lemma4}
For any bipartite state $\rho^{AB}$ and $\alpha \in [0,2] \backslash \{1\}$,
we have
\begin{align}
\label{yae4}
\ere{\alpha}(A:B)_\rho \geq \max \left\{ I_\alpha(A \rangle B)_\rho,
I_\alpha(B \rangle A)_\rho \right\}.
\end{align}
For a pure state $\ket{\Psi}^{AB}$, we have
\beq
\label{yae5}
S_{1/\alpha}(A)_\Psi \leq \ere{\alpha}(A:B)_\Psi \leq S_{2-\alpha}(A)_\Psi.
\enq
\end{lemma}
\begin{proof}

We first prove \eqref{yae4}.
We first note that for a separable state $\sigma^{AB} \in \sep^{A:B}$,
$\sigma^{AB} \leq \sigma^A \otimes \eye^B$.

For $\alpha \in (1,2]$, invoke the operator
monotonicity of $x \mapsto x^{\alpha-1}$,
and the fact that $A \geq B$ implies $A^{-1} \leq B^{-1}$ 
(replacing inverses by generalized inverses if the matrices are singular) to
have $(\sigma^{AB})^{1-\alpha} \geq (\sigma^A)^{1-\alpha} \otimes \eye^B$.
For $\alpha \in [0,1)$, invoke the operator monotonicity of $x \mapsto
x^\alpha$ to have
$(\sigma^{AB})^{1-\alpha} \leq (\sigma^A)^{1-\alpha} \otimes \eye^B$.

Using this, we have for $\alpha \in [0,2] \backslash \{1\}$,
\begin{align}
\ere{\alpha}(A:B)_\rho & \geq \inf_{\sigma^A}
\frac{1}{\alpha-1} \log \tr \left[ \tr_B (\rho^{AB})^\alpha \right]
(\sigma^{A})^{1-\alpha} \\
\label{yyae1}
& = \frac{\alpha}{\alpha-1} \log \tr
\left[ \tr_B (\rho^{AB})^\alpha \right]^{1/\alpha} \\
& = I_\alpha(A \rangle B),
\end{align}
where the first equality follows from the Sibson's identity (see the 
supplementary material of Ref. \cite{sharma-prl-2013}).
Arguing similarly by using $\sigma^{AB} \leq \eye^A \otimes \sigma^B$, we get
$\ere{\alpha}(A:B)_\rho \geq I_\alpha(B \rangle A)$.

We now prove \eqref{yae5}. For a pure state $\ket{\Psi}^{AB}$,
$\rho^A = \tr_B \ketbra{\Psi}^{AB}$, and using \eqref{yyae1}, we get
\beq
\ere{\alpha}(A:B)_\rho \geq \frac{\alpha}{\alpha-1} \log \tr (\rho^A)^{1/\alpha}
= S_{1/\alpha}(A)_\rho.
\enq
To prove the upper bound, choose
$\sigma^{AB} = \sum_i p_i \ketbra{i}^A \otimes \ketbra{i}^B$, where
\linebreak
$\ket{\Psi}^{AB} = \sum_i \sqrt{p_i} \ket{i}^A \ket{i}^B$ is the Schmidt 
decomposition.
\end{proof}

We now have the following inequality by van Dam and Hayden
\cite{dam-hayden-2008}.

\begin{lemma}[van Dam and Hayden \cite{dam-hayden-2008}]
\label{lemma11}
Let $F(\rho,\sigma) \geq F$. For $\alpha \in [0.5,1)$, the following holds:
\beq
S_\alpha(\rho) \geq S_\beta(\sigma) + \frac{2 \alpha}{1 - \alpha} \log F,
\enq
where $\beta = \infty$ if $\alpha = 0.5$ and $\beta = \alpha/(2 \alpha - 1)$ otherwise.
\end{lemma}
Using this lemma, we derive an inequality for the RREE.

\begin{lemma}
\label{lemma5}
Let $\rho^{AB}$ be a bipartite state and $\ket{\Psi}^{AB}$ be a
pure state such that
$F(\Psi^{AB},\rho^{AB}) \geq F$. Then for $\alpha \in (1,2]$, we have
\begin{align}
\label{dummy3}
\ere{\alpha}(A:B)_\rho & \geq \frac{2\alpha}{\alpha-1} \log F +
S_{\frac{1}{2-\alpha}}(A)_\Psi.
\end{align}
 \end{lemma}
\begin{proof}
We first note from \eqref{yae4} in Lemma \ref{lemma4} that
\beq
\label{yae6}
\ere{\alpha}(A:B)_\rho \geq I_\alpha(A \rangle B)_\rho.
\enq
Define
\beq
\sigma^{A} = \frac{ \tr_B (\rho^{AB})^\alpha }{ \tr (\rho^{AB})^\alpha }.
\enq
Then \eqref{yae6} can be written as
\beq
\label{yae7}
\ere{\alpha}(A:B)_\rho \geq S_{\frac{1}{\alpha}}(A)_\sigma
- S_{\alpha}(AB)_\rho.
\enq
We note that
\begin{align}
\log F(\sigma^A, \Psi^A) & \geq \log F(\sigma^{AB}, \Psi^{AB}) \\
& = \log F \left[ (\rho^{AB})^\alpha, \Psi^{AB} \right] - \frac{1}{2} \log \tr
(\rho^{AB})^\alpha \\
& \geq \alpha \log F(\rho^{AB}, \Psi^{AB})-\frac{1}{2} \log \tr (\rho^{AB})^\alpha \\
\label{yae8}
& \geq \alpha \log F - \frac{1}{2} \log \tr (\rho^{AB})^\alpha,
\end{align}
where the first inequality follows from the monotonicity under partial
trace, the equality follows since
$F(\rho,\sigma) = \tr \sqrt{\sqrt{\sigma} \rho \sqrt{\sigma}}$, the second
inequality follows since $\brackett{\Psi}{\rho^\alpha}{\Psi}^{AB}$
$\geq$ $(\brackett{\Psi}{\rho}{\Psi}^{AB})^\alpha$. Using \eqref{yae8},
we have
\beq
\label{yae9}
\frac{2}{\alpha-1} \log F(\sigma^A, \Psi^A) \geq
\frac{2 \alpha}{\alpha-1} \log F +  S_{\alpha}(AB)_\rho.
\enq
We now have from \eqref{yae7}
\begin{align}
\ere{\alpha}(A:B)_\rho & \geq S_{\frac{1}{\alpha}}(A)_\sigma
- S_{\alpha}(AB)_\rho \\
& \geq \frac{2}{\alpha-1} \log F(\sigma^A,\Psi^A) +
S_{\frac{1}{2-\alpha}}(A)_\Psi - S_{\alpha}(AB)_\rho \\
& \geq \frac{2 \alpha}{\alpha-1} \log F + S_{\frac{1}{2-\alpha}}(A)_\Psi,
\end{align}
where the second inequality follows from Lemma \ref{lemma11},
the third inequality follows from \eqref{yae9} and the claim follows.
\end{proof}

\section{Strong converse for  quantum state merging}

The following definition of the quantum state merging is almost
the same as in Ref. \cite{state-merging-2005}.

\begin{definition}[State-merging]
Let a pure state $\ket{\Psi}^{ABR}$ be shared between Alice ($A$) and
Bob ($B$). Let Alice and Bob have quantum registers $A_0$, $A_1$ and
$B_0$, $B_1$ respectively. A $(\Psi, F)$ {\bf state-merging} is a LOCC 
quantum operation
$\cM : A A_0 \otimes B B_0 \to A_1 \otimes B_1 B^\prime B$ such that
for $\rho^{A_1 B_1 B^\prime B R} =$
$\cM \left( \Psi^{ABR} \otimes \Phi_K^{A_0 B_0} \right)$,
\beq
F \left( \rho^{A_1 B_1 B^\prime B R}, \Phi_L^{A_1 B_1} \otimes
\Psi^{B^\prime B R} \right) \geq F.
\enq
The number $\log K - \log L$ is called the {\bf entanglement cost} of the 
protocol. In case of many copies $\Psi = \psi^{\otimes n}$,
$(\log K - \log L)/n$ is called the {\bf entanglement rate} of the protocol.
A real number $R$ is called an {\bf achievable rate} if there exist, for
$n \to \infty$, state-merging protocols of rate approaching $R$
and $F$ approaching $1$. The smallest achievable rate is the {\bf merging 
cost} of $\psi$.
\end{definition}

A fundamental result in quantum information theory is given in
the next theorem.

\begin{theorem}[Horodecki, Oppenheim and Winter \cite{state-merging-2005}]
For a state $\rho^{AB} = \tr_R \Psi^{ABR}$
shared by Alice and Bob, the merging cost is the
quantum conditional entropy $S(A|B)_\Psi$. If $S(A|B)_\Psi$ is positive, then
$R > S(A|B)_\Psi$ ebits are required per input copy and if it is negative, then
$R < -S(A|B)_\Psi$ ebits are obtained per input copy by the protocol.
\end{theorem}

A converse to the above theorem states that if $\log K - \log L < n S(A|B)_\Psi$,
then the fidelity would be bounded away from $1$.
A strong converse additionally
states under the same conditions that the fidelity would go to $0$ with $n$.

A weak converse is provided in Ref. \cite{state-merging-2005}. Strong
converse for the entanglement rate for this protocol can be construed
through the achievability and strong converse for the Schumacher compression
\cite{andreas-email-apr-2014}. 
Berta also provided a strong converse for the entanglement rate
in Ref. \cite{berta-2009}.

\subsection{Strong converse with LOCC}

We are now ready to provide the converse for the state-merging protocol.

\begin{theorem}[Strong converse for the entanglement rate]
For a $(\Psi,F)$ quantum state merging protocol, the following bound
holds for $\alpha \in (1,2]$
\beq
\log F \leq n \frac{\alpha-1}{2\alpha} \left[ \frac{\log K - \log L}{n} +
S_{2-\alpha}(B)_\Psi - S_{\frac{1}{2-\alpha}}(AB)_\Psi \right].
\enq
\end{theorem}
\begin{proof}
We have
\begin{align}
\ere{\alpha}(AA_0R:B_0B)_{\ket{\Psi}^{ABR} \otimes \ket{\Phi_K}^{A_0 B_0}}
& \leq S_{2-\alpha}(AA_0R)_{\ket{\Psi}^{ABR} \otimes
\ket{\Phi_K}^{A_0 B_0}} \\
& \leq \log K + S_{2-\alpha}(B)_{\ket{\Psi}}, 
\end{align}
where the first inequality follows from Lemma \ref{lemma4} and the second one from
the additivity of the \renyi entropies for the product states.
Let $\ket{\varphi}^{B^\prime B R A_1 B_1} =$
$\ket{\Psi}^{B^\prime B R} \otimes \ket{\Phi}^{A_1 B_1}$.
We now have
\begin{align}
\ere{\alpha}(A_1 R : B_1 B^\prime B)_\rho & \geq
\frac{2\alpha}{\alpha-1} \log F +
S_{\frac{1}{2-\alpha}}(B_1 B^\prime B)_{\ket{\varphi}} \\
& \geq \frac{2\alpha}{\alpha-1} \log F + \log L + S_{\frac{1}{2-\alpha}}
(AB)_{\ket{\Psi}^{ABR}},
\end{align}
where the first inequality follows from Lemma \ref{lemma5}.
Using the LOCC monotonicity of the RREE, we have
\beq
\ere{\alpha}(AA_0R:B_0B)_{\ket{\Psi}^{ABR}
\otimes \ket{\Phi_K}^{A_0 B_0}} \geq \ere{\alpha}(A_1 R : B_1 B^\prime B)_\rho,
\enq
and hence,
\beq
\log F \leq \frac{\alpha-1}{2\alpha} \left[ \log K - \log L + S_{2-\alpha}
(B)_{\ket{\Psi}} - S_{\frac{1}{2-\alpha}}(AB)_{\ket{\Psi}} \right].
\enq
For $n$ copies, we have
\beq
\log F \leq n \frac{\alpha-1}{2\alpha} \left[ \frac{\log K - \log L}{n} +
S_{2-\alpha}(B)_{\ket{\Psi}} -
S_{\frac{1}{2-\alpha}}(AB)_{\ket{\Psi}} \right].
\enq
Note that $S_{\frac{1}{2-\alpha}}(AB)_{\ket{\Psi}} -
S_{2-\alpha}(B)_{\ket{\Psi}}$ $\leq S(A|B)_{\ket{\Psi}}$ and one
can make the LHS approach the RHS by bringing $\alpha$ close to $1$
from above.
\end{proof}

\subsection{Strong converse with one-way LOCC}

As mentioned in Ref. \cite{state-merging-2005}, state-merging can be achieved with
just one-way LOCC. Note that since
one-way LOCC is a special case of LOCC, therefore, the converse for the entanglement
rate remains the same. What we need to provide is the converse for classical
communication cost.

\begin{theorem}[Strong converse for the classical communication cost]
For a $(\Psi,F)$ quantum state merging protocol, the following bound
holds for $\alpha \in (0.5,1)$ and $\beta = \alpha/(2 \alpha-1)$,
\beq
\log F \leq n \frac{1-\alpha}{4\alpha} \left[ \frac{\log |X|}{n} -
S_\beta(A)_\Psi - S_\beta(R)_\Psi + S_\alpha(AR)_\Psi \right].
\enq
\end{theorem}
\begin{proof}
Note that a one-way LOCC could be treated  as two cptp maps:
the first one by Alice $\cE: A A_0 \to X A_1$, where $X$ is a classical
register modeling the classical communication from Alice to Bob, and the
second one by Bob $\cD: XBB_0 \to B_1 B^\prime B$.

Let us assume that $\cE$ is constructed using an isometry $U_\cE: A A_0 \to
X A_1 E_1$, where $E_1$ is the environment and $\cD$ is constructed
using an isometry $U_\cD: X B B_0 \to B_1 B^\prime B E_2$,
where $E_2$ is the environment. Let (with some abuse of notation in
the order of the subsystems)
\begin{align}
\ket{\Xi}^{A_1 X B B_0 E_1} & = U_\cE \ket{\Psi}^{ABR}
\ket{\Phi_K}^{A_0 B_0} \\
\ket{\Omega}^{A_1 B_1 B^\prime B R E_1 E_2} & = U_\cD \circ U_\cE
\ket{\Psi}^{ABR}.
\end{align}
Since
\beq
F \left( \rho^{A_1 B_1 B^\prime B R}, \Phi_L^{A_1 B_1} \otimes
\Psi^{B^\prime B R} \right) \geq F,
\enq
and $\ket{\Omega}^{A_1 B_1 B^\prime B R E_1 E_2}$ is a purification of
$\rho^{A_1 B_1 B^\prime B R}$, using the Uhlmann's theorem
\cite{uhlmann-1976}, there exists a pure state
$\kappa^{E_1 E_2}$ such that
\beq
F ( \rho^{A_1 B_1 B^\prime B R}, \Phi_L^{A_1 B_1} \otimes
\Psi^{B^\prime B R}) =
F (\ket{\Omega}^{A_1 B_1 B^\prime B R E_1 E_2},
\Phi_L^{A_1 B_1} \otimes \Psi^{B^\prime B R} \otimes \kappa^{E_1 E_2} ).
\enq
Using the monotonicity of Fidelity under the partial trace, we have
$F(\rho^{A_1 E_1 R},  \Phi_L^{A_1}
\otimes \Psi^{R} \otimes \kappa^{E_1}) \geq F$. Now
invoking \eqref{dummy9}, we have
\beq
\label{dummy10}
F(\rho^{R A_1 E_1}, \Psi^R \otimes \rho^{A_1 E_1}) \geq F^2.
\enq
We now have
\begin{align}
\log |X| & \geq S_\beta(X A_1 E_1)_\Xi - S_\beta(A_1 E_1)_\Xi \\
& \geq S_\beta(X A_1 E_1)_\Xi - S_\beta(A_1 E_1)_\Xi - \nonumber \\
& \hspace{1in} \left[ S_\alpha(R X A_1 E_1)_\Xi -
S_\alpha(R A_1 E_1)_\Xi \right] \\
& = S_\beta(AA_0)_{\Psi \otimes \Phi_K} -
S_\alpha(RAA_0)_{\Psi \otimes \Phi_K} + \nonumber \\
& \hspace{1in} S_\alpha(RA_1 E_1)_\rho - S_\beta(A_1 E_1)_\rho \\
& = S_\beta(A)_\Psi - S_\alpha(RA)_\Psi + S_\alpha(RA_1 E_1)_\rho -
S_\beta(A_1 E_1)_\rho \\
& \geq S_\beta(A)_\Psi - S_\alpha(RA)_\Psi + S_\beta(R)_\Psi +
\frac{4 \alpha}{1-\alpha} \log F,
\end{align}
where the first two inequalities follows from Lemma \ref{lemma10}, the first 
equality follows since the \renyi entropies are invariant under isometries
(the isometry in question is $U_\cE$), the second equality follows by
canceling out $S_\beta(A_0)_{\Phi_K}$ and using the fact that
$S_\beta(A_0)_{\Phi_K} = S_\alpha(A_0)_{\Phi_K}$,
the third inequality follows from \eqref{dummy9} to have
\beq
S_\alpha(RA_1 E_1)_\rho \geq
S_\beta(RA_1E_1)_{\Psi^R \otimes \rho^{A_1 E_1}}
+ \frac{4\alpha}{1-\alpha} \log F,
\enq
where we have used the fact that $\rho^R = \Psi^R$.

Since $S_\beta(A)_\Psi + S_\beta(R)_\Psi - S_\alpha(RA)_\Psi$
$\leq I(A:R)_\Psi$ and LHS can be made to approach the RHS by choosing
$\alpha$ close to $1$ from below. Evaluating the above for $n$ copies
gives us the claim of the theorem.
\end{proof}

\section{Strong converse for entanglement concentration}

We now provide a short proof for a strong converse for the entanglement 
concentration protocol which is first defined below.

\begin{definition}[Entanglement concentration]
Let a pure state $\ket{\Psi}^{AB}$ be shared between Alice ($A$) and
Bob ($B$). A $(\Psi, F)$ {\bf entanglement concentration} protocol
is a LOCC quantum
operation $\cM : A \otimes B \to A_1 \otimes B_1$ such that
for $\rho^{A_1 B_1} = \cM \left( \Psi^{AB} \right)$,
\beq
F \left( \rho^{A_1 B_1}, \Phi_L^{A_1 B_1} \right) \geq F.
\enq
In case of many copies $\Psi = \psi^{\otimes n}$, $\log L/n$ is called
the rate of the protocol. A real number $R$ is called an
{\bf achievable rate} if there exist, for $n \to \infty$, entanglement concentration
protocols of rate approaching $R$ and fidelity approaching $1$.
\end{definition}

\begin{theorem}[Bennett \emph{et al} \cite{distill-1996}]
For a pure state $\Psi^{AB}$ shared by Alice and Bob,
the highest achievable rate is given by $S(A)_\Psi$.
Conversely, any concentration protocol achieving rates higher than
$S(A)_\rho$ would have fidelity bounded away from $1$.
\end{theorem}

The strong converse for this protocol follows from the results in
Ref. \cite{datta-leditzky-2014}.
We now prove a strong converse.

\begin{theorem}[Strong converse for entanglement concentration]
For any $(\Psi,F)$ entanglement concentration protocol, the following bound
holds for $\alpha \in (1,2]$
\beq
\log F \leq n \frac{\alpha-1}{2\alpha} \left[ S_{2-\alpha}(A)_\Psi -
\frac{\log L}{n} \right].
\enq
\end{theorem}
\begin{proof}
We have
\beq
n S_{2-\alpha}(A)_\Psi \geq \ere{\alpha}(A:B)_\Psi \geq
\ere{\alpha}(A_1:B_1)_{\rho}
\geq \frac{2\alpha}{\alpha-1} \log F + \log L,
\enq
where the first inequality follows from Lemma \ref{lemma4}, the second 
inequality follows
from Lemma \ref{lemma1} and the third inequality follows from Lemma 
\ref{lemma5}.
\end{proof}

\section{Strong converse for Schumacher compression}

We now provide a short proof for a strong converse for the Schumacher 
compression protocol which is first defined below.

\begin{definition}[Schumacher compression]
Let Alice have a quantum state $\rho^{A}$ with purification
$\Psi^{RA}$ that she wants to transfer to Bob.
A $(\rho, F)$ {\bf Schumacher compression} protocol consists of a cptp 
compression operation $\cC : A \to B$ (with isometry $U_{\cC} : A \to B E_1$)
and a cptp decompression operation $\cD : B \to A_1$
(with isometry $U_{\cD} : B \to A_1 E_2$) with
$\ket{\Omega}^{R A_1 E_1 E_2} = U_{\cD} \circ U_{\cC} \ket{\Psi}^{RA}$
such that
\beq
F \left( \Omega^{R A_1}, \Psi^{R A_1} \right) \geq F.
\enq
In case of many copies $\rho^{\otimes n}$, $\log |B|/n$ is called
the rate of the protocol. A real number $R$ is called an
{\bf achievable rate} if there exist, for $n \to \infty$, Schumacher
compression protocols of rate approaching $R$ and fidelity
approaching $1$.
\end{definition}

\begin{theorem}[Schumacher \cite{schu-noiseless-1995}]
The smallest achievable rate for Schumacher compression is given
by $S(A)_\rho$.
Conversely, any Schumacher compression protocol achieving rates
smaller than $S(A)_\rho$ would have fidelity bounded away from $1$.
\end{theorem}

A strong converse for this protocol
was first provided by Winter in Ref. \cite{andreas-thesis} and
also follows from Ref. \cite{datta-leditzky-2014}.
We now prove another strong converse.

\begin{theorem}[Strong converse for Schumacher compression]
For any $(\rho,F)$ Schumacher compression protocol, the following bound
holds for $\alpha \in (0.5,1)$ and $\beta = \alpha/(2 \alpha - 1)$
\beq
\log F \leq n \frac{1-\alpha}{2\alpha} \left[ \frac{\log|B|}{n} -
S_{\beta}(A)_\rho \right].
\enq
\end{theorem}
\begin{proof}
Using Uhlmann's theorem \cite{uhlmann-1976}, there exists a pure state
$\tau^{E_1 E_2}$ such that \linebreak
$F(\Omega^{R A_1 E_1 E_2}, \Psi^{R A_1}
\otimes \tau^{E_1 E_2}) \geq F$. Using the monotonicity of the Fidelity
under partial trace, we get
$F(\Omega^{R E_1}, \Psi^R \otimes \tau^{E_1}) \geq F$. We have
\begin{align}
\log |B| \geq S_\alpha(B)_\Omega = S_\alpha(R E_1)_\Omega
\geq S_\alpha(RE_1)_\Omega - S_\beta(E_1)_\tau \\
\geq S_\beta(R)_\Psi + \frac{2 \alpha}{1 - \alpha} \log F \\
= S_\beta(A)_\rho + \frac{2 \alpha}{1 - \alpha} \log F,
\end{align}
where in the first inequality comes from the fact that a maximally mixed
state gives the maximum \renyi entropy, the first equality is because
the \renyi entropy does not change under the application of an isometry,
the second inequality follows since we are subtracting a non-negative
term, the last inequality follows from Lemma \ref{lemma11} and
the last equality follows since $\Psi^{RA}$ is pure.
Evaluating for $n$ copies $\rho^{\otimes n}$ gives us the
claim of the theorem.
\end{proof}

\section{Conclusions and Acknowledgements}

To conclude, inspired by the Polyanskiy-Verd\'{u} paradigm, we provide
strong converses using generalized divergences for protocols
involving LOCC maps. While this is illustrated for two protocols
using LOCC, nevertheless, the ideas are presented in generality
without undue restrictions or
made specific to a particular protocol. We provide some inequalities involving
the \renyi relative entropy based quantities, in particular, the \renyi relative
entropy of entanglement.

We then provide a strong converse for the
quantum state merging protocol both for the entanglement rate as well
as for the classical communication cost (for one-way LOCC).
We also provide a short proof of a strong converse for entanglement 
concentration of pure states leveraging
the earlier developments and for Schumacher compression.
It may be possible to improve upon the exponents that we provided.

The author thanks Andreas Winter and Mark Wilde for their comments
and to Andreas Winter for mentioning a strong converse for
the entanglement rate for
state merging using the Schumacher's achievability and strong converse,
and for pointing an error in an earlier version in the interpretation
of the bounds obtained for state merging.

\appendix

\section{Some more inequalities}

The following lemmas are used in the text. The only exception is
Lemma \ref{lemma-not-needed-yet} which was derived for an approach
that got nowhere. We state it here nevertheless.

\begin{lemma}
\label{lemma-not-needed-yet}
For a cq state $\rho^{RX} = \sum_x p_x \rho_x \otimes \ketbra{x}^X$,
$\alpha > 1$, where $\rho_x$ are density matrices
and $\{ p_x \}$ a probability vector, we have
\begin{align}
\log |X| & \geq S_\alpha(R)_\rho - S_\alpha(R|X)_\rho,
\end{align}
where $S_\alpha(R|X)_\rho \equiv - I_\alpha(X \rangle R)$.
\end{lemma}
\begin{proof}
We have to show that
\beq
\sum_x p_x \left( \tr \rho_x^\alpha \right)^{1/\alpha} \leq
|X|^{(\alpha-1)/\alpha} \, \left[ \tr \left( \sum_x p_x \rho_x \right)^\alpha \right]^{1/\alpha}.
\enq
We now have
\begin{align}
\text{RHS} & = \left[ |X|^{\alpha-1} \tr \left( \sum_x p_x \rho_x \right)^\alpha \right]^{1/\alpha}  \\
& \geq \left( |X|^{\alpha-1} \sum_x p_x^\alpha \tr \rho_x^\alpha \right)^{1/\alpha}  \\
& = \left( \sum_x \frac{1}{|X|} |X|^{\alpha} p_x^\alpha \tr \rho_x^\alpha \right)^{1/\alpha}  \\
& \geq \sum_x \frac{1}{|X|} |X| p_x \left( \tr \rho_x^\alpha \right)^{1/\alpha} \\
& = \text{LHS},
\end{align}
where the first inequality follows from the easy to prove inequality
(see Ref. \cite{simon-book}
for example) that for positive operators $\rho$ and $\sigma$ and
$\alpha \geq 1$,
$\tr (\rho + \sigma)^\alpha \geq \tr \rho^\alpha + \tr \sigma^\alpha$ and the 
second inequality follows from the concavity of $x \mapsto x^{1/\alpha}$.
\end{proof}

\begin{lemma}
\label{lemma10}
For a cq state $\rho^{RX}$ (classical in $X$),
$\alpha \in [0,2] \backslash \{1\}$, we have
\begin{align}
\label{dummy6}
\log |X| & \geq S_\alpha(RX)_\rho - S_\alpha(R)_\rho \geq 0.
\end{align}
\end{lemma}
\begin{proof}
Let
\beq
\rho^{RX} = \sum_x p_x \rho_x \otimes \ketbra{x}^X.
\enq
We now have
\begin{align}
S_\alpha(R)_\rho & = -\frac{1}{\alpha-1} \log \tr \left( \sum_x p_x \rho_x \right)^\alpha \\
S_\alpha(RX)_\rho & = -\frac{1}{\alpha-1} \log \sum_x p_x^\alpha \tr \rho_x^\alpha.
\end{align}
Let $\alpha \in (1,2)$. To show the left inequality of \eqref{dummy6}, we
have to show that
\beq
|X|^{\alpha-1} \sum_x p_x^\alpha \tr \rho_x^\alpha \geq \tr \left( \sum_x
p_x \rho_x \right)^\alpha.
\enq
We now have
\begin{align}
\text{RHS} & = \tr \left( \sum_x \frac{1}{|X|} |X| p_x \rho_x \right)^\alpha \\
& \leq \sum_x \frac{1}{|X|} |X|^\alpha p_x^\alpha \tr \rho_x ^\alpha \\
& = \text{LHS},
\end{align}
where the inequality follows from the operator convexity of $x \mapsto x^{\alpha}$.

To show the right inequality of \eqref{dummy6}, we have to show that
\beq
\tr \left( \sum_x p_x \rho_x \right)^\alpha \geq \sum_x p_x^\alpha \tr
\rho_x^\alpha,
\enq
which follows easily \cite{simon-book}.

The case of $\alpha < 1$ is treated similarly.
\end{proof}

\begin{lemma}
For density matrices $\rho^{AB}$, $\tau^A$, $\sigma^B$, and
$\rho^B = \tr_A \rho^{AB}$, we have
\beq
\label{dummy7}
F(\rho^{AB}, \tau^A \otimes \rho^B) \geq \left[ F(\rho^{AB}, \tau^A \otimes \sigma^B) \right]^2.
\enq
\end{lemma}
\begin{proof}
Let us first prove for a pure state $\rho^{AB} = \ketbra{\Psi}^{AB}$.
Let $\Psi^B = \tr_A \Psi^{AB}$ and we have to show that
\beq
\label{dummy8}
F(\Psi^{AB}, \tau^A \otimes \Psi^B) \geq \left[
F(\Psi^{AB}, \tau^A \otimes \sigma^B) \right]^2.
\enq
Let the Schmidt decomposition be $\ket{\Psi}^{AB} =$
$\sum_i \sqrt{\lambda_i} \ket{i}^A \ket{i}^B$. Let $\Psi^A = \tr_B \Psi^{AB}$. Then
\linebreak $F(\Psi^{AB}, \tau^A \otimes \Psi^B) =
\sqrt{\tr \, \tau^A (\Psi^A)^2}$ and
\begin{align}
\left[ F(\Psi^{AB}, \tau^A \otimes \sigma^B) \right]^2 & =
\tr \Psi^{AB} (\tau^A \otimes \sigma^B) \\
& \leq \tr \Psi^{AB} ( \tau^A \otimes \eye^B) \\
& = \tr \Psi^A \tau^A \\
& \leq \sqrt{\tr \, \tau^A (\Psi^A)^2} \\
& = F(\Psi^{AB}, \tau^A \otimes \Psi^B),
\end{align}
where the first inequality follows since $\sigma^B \leq \eye^B$ and the
second inequality is easy to prove.

Now let us prove \eqref{dummy7} for mixed $\rho^{AB}$.
Let $\Psi^{R_1 A B R_2}$,
$\tau^{R_1 A}$, $\sigma^{B R_2}$ be purifications of $\rho^{AB}$, $\tau^A$ 
and $\sigma^B$ respectively such that (invoking the Uhlmann's Theorem 
\cite{uhlmann-1976})
\beq
F(\rho^{AB}, \tau^A \otimes \sigma^B) = F(\Psi^{R_1 A B R_2}, \tau^{R_1 A} \otimes
\sigma^{B R_2}).
\enq
Let $\Psi^{B R_2} = \tr_{R_1 A} \Psi^{R_1 A B R_2}$  and we make no
assumption that it is pure. We now have
\begin{align}
F(\rho^{AB}, \tau^A \otimes \rho^B) & \geq F(\Psi^{R_1 A B R_2},
\tau^{R_1 A} \otimes \Psi^{B R_2}) \\
& \geq \left[ F(\Psi^{R_1 A B R_2}, \tau^{R_1 A} \otimes \sigma^{B R_2}) 
\right]^2 \\
& =  \left[ F(\rho^{AB}, \tau^A \otimes \sigma^B) \right]^2,
\end{align}
where the first inequality follows from the monotonicity of the fidelity under 
partial trace, the second inequality follows from \eqref{dummy8}, and the 
equality at the end follows from our choice of purifications.

In particular, choosing $\tau^A = \tr_B \rho^{AB} \equiv \rho^A$ yields
\beq
\label{dummy9}
F(\rho^{AB}, \rho^A \otimes \rho^B) \geq
\left[ F(\rho^{AB}, \rho^A \otimes \sigma^B) \right]^2.
\enq
This inequality would be put to use later. QED.
\end{proof}


\begin{thebibliography}{10}
\providecommand{\url}[1]{\texttt{#1}}
\providecommand{\urlprefix}{URL }

\bibitem{covertom}
T.~M. Cover and J.~A. Thomas.
\newblock \emph{Elements of {I}nformation {T}heory}.
\newblock Wiley, Hoboken, NJ, USA, 2nd edn.,
\newblock 2006.

\bibitem{wilde-book}
M.~M. Wilde.
\newblock \emph{Quantum Information Theory}.
\newblock Cambridge University Press,
\newblock 2013.

\bibitem{hayashi}
{M. Hayashi}.
\newblock \emph{Quantum Information: An Introduction}.
\newblock Springer, Berlin,
\newblock 2006.

\bibitem{petz-book}
{M. Ohya} and {D. Petz}.
\newblock \emph{Quantum {E}ntropy and its use}.
\newblock Springer-Verlag, Berlin, 1st edn.,
\newblock 1993.

\bibitem{wolfowitz-book}
{J. Wolfowitz}.
\newblock \emph{Coding Theorems of Information Theory}.
\newblock Prentice-Hall, Englewood Cliffs, NJ, USA,
\newblock 1962.

\bibitem{arimoto-1973-converse}
{S. Arimoto}.
\newblock \href{http://dx.doi.org/10.1109/TIT.1973.1055007}{\emph{{On the
  converse to the coding theorem for discrete memoryless channels}}}.
\newblock \emph{IEEE Trans. Inf. Theory}, vol.~19: pp. 357 -- 359,
\newblock May 1973.

\bibitem{winter-99-converse}
{A. Winter}.
\newblock \href{http://dx.doi.org/10.1109/18.796385}{\emph{Coding theorem and
  strong converse for quantum channels}}.
\newblock \emph{IEEE Trans. Inf. Theory}, vol.~45: pp. 2481--2485,
\newblock Nov. 1999.

\bibitem{andreas-thesis}
{A. Winter}.
\newblock \href{http://arxiv.org/abs/quant-ph/9907077}{\emph{Coding theorems of
  quantum information theory}}.
\newblock \emph{arXiv:quant-ph/9907077},
\newblock 1999.

\bibitem{ogawa-1999-converse}
T.~Ogawa and H.~Nagaoka.
\newblock \href{http://dx.doi.org/10.1109/18.796386}{\emph{Strong converse to
  the quantum channel coding theorem}}.
\newblock \emph{IEEE Trans. Inf. Theory}, vol.~45: pp. 2486--2489,
\newblock Nov. 1999.

\bibitem{konig-2009-converse}
{R. K\"{o}nig} and {S. Wehner}.
\newblock \href{http://dx.doi.org/10.1103/PhysRevLett.103.070504}{\emph{A
  strong converse for classical channel coding using}}
  \href{http://dx.doi.org/10.1103/PhysRevLett.103.070504}{\emph{entangled inputs}}.
\newblock \emph{Phys. Rev. Lett.}, vol. 103: p. 070504,
\newblock Aug. 2009.

\bibitem{renner-wolf-2004}
R.~Renner and S.~Wolf.
\newblock \href{http://dx.doi.org/10.1109/ISIT.2004.1365269}{\emph{Smooth
  {R}\'{e}nyi entropy and applications}}.
\newblock In \emph{Proc. IEEE Int. Symp. Inf. Theory (ISIT)}. Chicago, IL, USA,
\newblock June 2004.

\bibitem{polyanskiy-2010-converse}
{Y. Polyanskiy} and {S. Verd\'{u}}.
\newblock \href{http://dx.doi.org/10.1109/ALLERTON.2010.5707067}{\emph{Arimoto
  channel coding converse and {R}\'{e}nyi divergence}}.
\newblock In \emph{Proc. 48th Allerton Conf. Comm. Cont. Comp.} Monticello, IL,
  USA,
\newblock Sept. 2010.

\bibitem{blahut1976}
{R. E. Blahut}.
\newblock \href{http://dx.doi.org/10.1109/TIT.1976.1055576}{\emph{Information
  bounds of the {F}ano-{K}ullback type}}.
\newblock \emph{IEEE Trans. Inf. Theory}, vol.~22: pp. 410--421,
\newblock July 1976.

\bibitem{sharma-prl-2013}
{N. Sharma} and {N. A. Warsi}.
\newblock
  \href{http://dx.doi.org/10.1103/PhysRevLett.110.080501}{\emph{Fundamental
  bound on the reliability of quantum}}
  \href{http://dx.doi.org/10.1103/PhysRevLett.110.080501}{\emph{information transmission}}.
\newblock \emph{Phys. Rev. Lett.}, vol. 110: p. 080501,
\newblock Feb. 2013.

\bibitem{wilde-2013}
{M. M. Wilde}, {A. Winter}, and {D. Yang}.
\newblock \href{http://arxiv.org/abs/1306.1586}{\emph{Strong converse for the
  classical capacity of}}
  \href{http://arxiv.org/abs/1306.1586}{\emph{entanglement-breaking and {H}adamard channels}}.
\newblock \emph{arXiv:1306.1586},
\newblock 2013.

\bibitem{gupta-wilde-2013}
{M. K. Gupta} and {M. M. Wilde}.
\newblock \href{http://arxiv.org/abs/1310.7028}{\emph{Multiplicativity of
  completely bounded $p$-norms}}
  \href{http://arxiv.org/abs/1310.7028}{\emph{implies a strong converse for
  entanglement-assisted capacity}}.
\newblock \emph{arXiv:1310.7028},
\newblock 2013.

\bibitem{no-strong-converse-2011}
{T. Dorlas} and {C. Morgan}.
\newblock
  \href{http://dx.doi.org/http://link.aps.org/doi/10.1103/PhysRevA.84.042318}{\emph{The
  invalidity of a strong capacity for a quantum channel}}
  \href{http://dx.doi.org/http://link.aps.org/doi/10.1103/PhysRevA.84.042318}{\emph{with memory}}.
\newblock \emph{Phys. Rev. A}, vol.~84: p. 042318,
\newblock Oct. 2011.

\bibitem{morgan-winter-2014}
{C. Morgan} and {A. Winter}.
\newblock \href{http://dx.doi.org/10.1109/TIT.2013.2288971}{\emph{Pretty strong
  converse for the quantum capacity of}}
  \href{http://dx.doi.org/10.1109/TIT.2013.2288971}{\emph{degradable channels}}.
\newblock \emph{IEEE Trans. Inf. Theory}, vol.~60: pp. 317--333,
\newblock Jan. 2014.

\bibitem{vidal-2000}
{G. Vidal}.
\newblock \href{http://dx.doi.org/10.1080/09500340008244048}{\emph{Entanglement
  monotones}}.
\newblock \emph{J. Mod. Opt.}, vol.~47: pp. 355--376,
\newblock 2000.

\bibitem{entanglement-horo-2009}
{R. Horodecki}, {P. Horodecki}, {M. Horodecki}, and {K. Horodecki}.
\newblock \href{http://dx.doi.org/10.1103/RevModPhys.81.865}{\emph{Quantum
  entanglement}}.
\newblock \emph{Rev. Mod. Phys.}, vol.~81: pp. 865--942,
\newblock June 2009.

\bibitem{vedral-1997}
{V. Vedral}, {M. B. Plenio}, {M. A. Rippin}, and {P. L. Knight}.
\newblock
  \href{http://dx.doi.org/10.1103/PhysRevLett.78.2275}{\emph{Quantifying
  entanglement}}.
\newblock \emph{Phys. Rev. Lett.}, vol.~78: pp. 2275--2279,
\newblock Mar. 1997.

\bibitem{dam-hayden-2008}
{W. van Dam} and {P. Hayden}.
\newblock \href{http://arxiv.org/abs/quant-ph/0204093}{\emph{R\'{e}nyi-entropic
  bounds on quantum communication}}.
\newblock
\newblock \emph{arXiv:quant-ph/0204093}.

\bibitem{petz-quasi-entr-1986}
{D. Petz}.
\newblock
  \href{http://dx.doi.org/10.1016/0034-4877(86)90067-4}{\emph{Quasi-entropies
  for finite quantum systems}}.
\newblock \emph{{Rep. Math. Phys.}}, vol.~23: pp. 57--65,
\newblock Feb. 1986.

\bibitem{petz-quasi-entr-2010}
{D. Petz}.
\newblock \href{http://dx.doi.org/10.3390/e12030304}{\emph{From $f$-divergence
  to quantum quasi-entropies and their use}}.
\newblock \emph{Entropy}, vol.~12: pp. 304--325,
\newblock Mar. 2010.

\bibitem{christandl-winter-2004}
{M. Christandl} and {A. Winter}.
\newblock \href{http://dx.doi.org/10.1063/1.1643788}{\emph{``{S}quashed
  entanglement": {An} additive entanglement}}
  \href{http://dx.doi.org/10.1063/1.1643788}{\emph{measure}}.
\newblock \emph{J. Math. Phys.}, vol.~45: pp. 829--840,
\newblock Mar. 2004.

\bibitem{petz-statistics-book}
{D\'{e}nes Petz}.
\newblock \emph{Quantum Information Theory and Quantum Statistics}.
\newblock Springer-Verlag, Berlin,
\newblock 2008.

\bibitem{state-merging-2005}
{M. Horodecki}, {J. Oppenheim}, and {A. Winter}.
\newblock \href{http://dx.doi.org/10.1007/s00220-006-0118-x}{\emph{Quantum
  state merging and negative}}
  \href{http://dx.doi.org/10.1007/s00220-006-0118-x}{\emph{ information}}.
\newblock \emph{Commun. Math. Phys.}, vol. 269: pp. 107--136,
\newblock 2007.

\bibitem{andreas-email-apr-2014}
{A. Winter}.
\newblock \emph{Personal communication},
\newblock Apr. 2014.

\bibitem{berta-2009}
{M. Berta}.
\newblock \href{http://arxiv.org/abs/0912.4495}{\emph{Single-shot quantum state
  merging}}.
\newblock \emph{M.S. Thesis, {ETH} {Zurich}, arXiv:0912.4495},
\newblock 2008.

\bibitem{uhlmann-1976}
{A. Uhlmann}.
\newblock \href{http://dx.doi.org/10.1016/0034-4877(76)90060-4}{\emph{The
  `transition probability' in the state space of a $\ast$-algebra}}.
\newblock \emph{{Rep. Math. Phys.}}, vol.~9: pp. 273--279,
\newblock 1976.

\bibitem{distill-1996}
{C. H. Bennett}, {H. Bernstein}, {S. Popescu}, and {B. Schumacher}.
\newblock \href{http://dx.doi.org/10.1103/PhysRevA.53.2046}{\emph{Concentrating
  partial entanglement by local operations}}.
\newblock \emph{Phys. Rev. A}, vol.~53: pp. 2046--2052,
\newblock April 1996.

\bibitem{datta-leditzky-2014}
{N. Datta} and {F. Leditzky}.
\newblock \href{http://arxiv.org/abs/1403.2543}{\emph{Second-order asymptotics
  for source coding, dense coding}}
  \href{http://arxiv.org/abs/1403.2543}{\emph{and pure-state entanglement conversions}}.
\newblock \emph{arXiv: 1403.2543},
\newblock 2014.

\bibitem{schu-noiseless-1995}
{B. Schumacher}.
\newblock \href{http://link.aps.org/doi/10.1103/PhysRevA.51.2738}{\emph{Quantum
  coding}}.
\newblock \emph{Phys. Rev. A}, vol.~51: pp. 2738--2747,
\newblock 1995.

\bibitem{simon-book}
{B. Simon}.
\newblock \emph{Trace ideals and their applications}.
\newblock London Math. Soc. {L}ecture note series 35. Cambridge University
  Press, Cambridge, Cambridge, U.K.,
\newblock 1979.

\end{thebibliography}
\end{document}